\documentclass[12pt]{amsart}
\usepackage{bbm}
\usepackage[english]{babel}
\usepackage{color}
\usepackage{anysize}
\usepackage{geometry}
\marginsize{2.0cm}{2.0cm}{2.0cm}{2.0cm}
\usepackage{hyperref}
\usepackage{url}
\usepackage[normalem]{ulem}
\usepackage{hyperref}
\usepackage{mathtools}
\mathtoolsset{showonlyrefs} 
\usepackage{amssymb}
\usepackage{tikz}
\usetikzlibrary{matrix,backgrounds}
\theoremstyle{plain}
\numberwithin{equation}{section}
\newtheorem{theorem}{Theorem}[section]
\newtheorem{corollary}[theorem]{Corollary}
\newtheorem{lemma}[theorem]{Lemma}
\newtheorem{remark}[theorem]{Remark}
\newcommand{\ud}{\mathrm{d}}
\newcommand{\ra}{\rightarrow}
\newcommand{\e}{\varepsilon}
\newcommand{\PP}{\mathbb{P}}
\newcommand{\NN}{\mathbb{N}}
\newcommand{\cC}{\mathcal{C}}
\newcommand{\gG}{\mathcal{G}}
\newcommand{\RR}{\mathbb{R}}
\newcommand{\EE}{\mathbb{E}}
\newcommand{\wW}{\mathcal{W}}
\newcommand{\lqq}{\leqslant}
\newcommand{\gqq}{\geqslant}
\newcommand{\cI}{\mathcal{I}}
\newcommand{\cR}{\mathcal{R}}
\newcommand{\cS}{\mathcal{S}}

\title[Cutoff thermalization for an infinite-dimensional shell model]
{Cutoff ergodicity bounds in Wasserstein distance\\
for a viscous energy shell model with L\'evy noise}
\author{G. Barrera}
\address{University of Helsinki, Department of Mathematical and Statistical Sciences. 
Exactum in Kumpula Campus. PL~68, Pietari Kalmin katu 5. 00560 Helsinki, Finland.}
\email{gerardo.barreravargas@helsinki.fi}
\author{M. A. H\"ogele}
\address{Universidad de los Andes, Departamento de Matem\'aticas, Bogot\'a, Colombia.}
\email{ma.hoegele@uniandes.edu.co}
\author{J. C. Pardo}
\address{
CIMAT. Jalisco S/N, Valenciana, CP. 36240. Guanajuato, Guanajuato, M\'exico.}
\email{jcpardo@cimat.mx}
\author{I. Pavlyukevich}
\address{
Friedrich Schiller University Jena, Department of Mathematics and Computer Science, Ernst-Abbe-Platz~2, 07743 Jena, Germany.
}
\email{ilya.palvyukevich@uni-jena.de}


\begin{document}
\begin{abstract}
This article establishes explicit non-asymptotic ergodic bounds in the renormalized Wasserstein-Kantorovich-Rubinstein (WKR) distance for a viscous energy shell lattice model of turbulence with random energy injection. 
The system under consideration is driven either by a Brownian motion, a symmetric $\alpha$-stable L\'evy process, a stationary Gaussian or $\alpha$-stable Ornstein-Uhlenbeck process, or by a general L\'evy process with second moments. The obtained non-asymptotic bounds establish asymptotically abrupt thermalization. 
The analysis is based on the explicit representation of the solution of the system in terms of convolutions of Bessel functions. 
\end{abstract}

\maketitle

\section{\textbf{Introduction}}

\noindent
Fully developed, isotropic turbulence is commonly understood as the zero viscosity limit of solutions of the Navier-Stokes equations. Since its beginnings in~\cite{Frisch95,Ko41} more and more elements of turbulence have been discovered, however, a unified approach remains missing, since its phenomenology involves large ranges of quantities over many scales of magnitude, which is morally related to self-similarity of the solutions of the idealized Euler equation. 

\noindent
In practice, it is paramount to limit the resulting computational cost of the simulation of turbulent phenomena in the context of aerodynamics and hydrodynamics such as whether forecasts by different types of model reductions, see~\cite{Di11}. A widely accepted class of reduced models of turbulence are the so-called (stochastic) shell models, i.e.,~complex-valued Fourier mode equations with a (possibly random) energy injection in lower modes and an energy transport to higher and higher modes, by a multiplicative (nonlinear) nearest-neighbor interaction of each node. The most studied shell models are the GOY model (after Glatzer, Ohkitani, Yamada,~\cite{Gl73, OY87}) and the SABRA model~\cite{LPPPV98}. 
Their random dynamics (well-posedness in correctly weighted Fourier sequence spaces, the existence and finite dimensionality of random attractors, large deviations principles and the existence and uniqueness of invariant measures) have been studied successfully~\cite{BBF08, BLPT07, BF09, BFT10, BGAS16, BM09, MM14}. These works fall into the larger class of lattice systems, see for instance~\cite{BLL06, CHSV16, HSZ15} and the references therein.  Recently, in~\cite{BER16} the authors show ergodicity and the strong Feller smoothing property of the laws for GOY and SABRA subject to L\'evy perturbations. The variational techniques used there provide exponential upper bounds only for large initial values, however, do not allow for sharp upper bounds of the rate of convergence, and virtually nothing is known about lower bounds. In general, the study of sharp bounds is a hard problem and requires completely different methods. 

\noindent
The study of asymptotically sharp upper and lower ergodic error bounds along a particular time scale can be often associated to the so-called cutoff phenomenon or abrupt thermalization, that is, the existence of a deterministic critical time scale $\tau$, which typically separates sharply ``small'' error values ahead of $\tau$, that is, $\tau+r$ for $r\gg 1$, and ``large'' error values for times lagging behind $\tau$, that is $\tau+r$ for $r\ll -1$. This concept was first introduced in the discrete context of (random) card shuffling and random walks on groups~\cite{AD, AD87, DI87, DIA96, DS}, where the distance between is taken to be the total variation distance. 
The cutoff phenomenon or abrupt thermalization is a very active field of mathematical research~\cite{Al83,BDJ2023,BJ,BY1,BBF08, BLY06,BY2014, BASU, BD92,B-HLP19, BCS19,BCS18,BCL22, CSC08, DGM90, Goncalves, HS20, LL19, Lachaud2005,  La16, LanciaNardiScoppola12, LLP10, LPW,  LubetzkySly13,  Meliot14, Scarabotti, Yc99}. In the physics literature this phenomenon has received growing attention with applications in different contexts: (quantum) Markov chains~\cite{Kastoryano12} and quantum information 
processing~\cite{Kastoryano13}, dissipative quantum 
circuits~\cite{JohnssonTicozziViola17,Oh}, 
Fermionic systems~\cite{Vernier20}, chemical kinetics~\cite{BOK10}, statistical mechanics~\cite{LubetzkySly13}, even deterministic systems such as coagulation-fragmentation equations 
in~\cite{Pego16,Pego17} and ferromagnetic spin models~\cite{CS21}. 

\noindent
We stress that in the continuous state space context, however, the total variation is not always suitable, since the associated topology on the space of probability distributions is too strong for many practical purposes. In particular, it is not continuous for discrete approximations of absolutely continuous distributions. For the respective counterexample see~\cite[Subsection~1.3.5]{BHPTV}. The most illustrating consequence of this drawback of the total variation distance is that the  celebrated central limit theorem of DeMoivre-Laplace is not valid in the total variation distance. A much more tractable distance between probability laws is given by the Wasserstein-Kantorovich-Rubinstein distance, which is based on the optimal transport (or coupling) between two given distributions, see~\cite{Panaretos,Vi09}. In~\cite{FKN21}, for instance, the authors study (abstract) WKR perturbations of Markovian transition semi-groups from a more analytical perspective. 
In~\cite{BHPWA, BHPWANO,BP} the authors studied linear and nonlinear Langevin equation subject to small noise in the WKR distance. In~\cite[Lemma 2.2(d)]{BHPWA} they establish the non-standard so-called shift linearity property of WKR distances of order $p\gqq 1$ in some Banach space $(B,|\cdot|)$, which additionally simplifies the calculations in the WKR distance: 
\[
\wW_p(u+X, X) = |u| \quad \textrm{ for all }\quad p\gqq 1,\, u\in B  
\]
and any $B$-valued random vector $X$ with finite $p$-th moment, $\EE[|X|^p]<\infty$. 
For lower bounds we use the general inequality 
\[
\wW_p(u+X,Y)\gqq |u| \quad \textrm{ for all }\quad p\gqq 1,\, u\in B  
\]
under the matching condition $\mathbb{E}[X]=\mathbb{E}[Y]$ whenever the corresponding $p$-th moments exist in the spirit of~\cite[Theorem~2.1]{Chen}.

\noindent
Due to the rich and mathematically challenging behavior of nonlinear systems like GOY and SABRA, in~\cite{MSV07, MSV07b} these models have been further reduced to infinite linear chains of oscillators with dissipation. In this article, we study a particular model of this class. In~\cite{PS08} the solution and its invariant measure of such systems have been calculated explicitly in terms of Bessel functions of the first kind. Even such extremely conceptualized and explicitly solvable models provide interesting insights, as eloquently put forward in the introduction of~\cite{MSV07}. 

\noindent
The main idea of this article is to combine the above mentioned (and other) advantages of the WKR distance with the explicit solvability of the equation in terms of stochastic (convolution) integrals over well-known special functions. In particular, they are based on particular coupling (replica) techniques between the current state of the system starting in $0$ and the limiting measures, and the detailed knowledge of the linear dynamics, in particular, the characteristics of the invariant measure, which is dominated by the sequence of Bessel functions of the first kind. 

\noindent
Here, we consider Gaussian white (Theorem~\ref{thm:cutoffconvergence}) and  red (Ornstein-Uhlenbeck) noises (Theorem~\ref{thm:convergenciaroja}), as well as $\alpha$-stable  (Theorem~\ref{thm:convergenciaestable}) and $\alpha$-stable Ornstein-Uhlenbeck noises (Theorem~\ref{thm:convergenciarojaestable}), 
as well as for general L\'evy noise with second moments (Theorem~\ref{thm:convergenciasegundo}). 
All results imply respective small noise results (Corollary~\ref{cor:small}, Remark~\ref{rem:smallGOU} and the respective remarks).

\noindent
The manuscript is organized as follows: In Section~\ref{s:setup} we expose our setting and give all necessary notation. In Section~\ref{s:Gaussianoblanco} we show non-asymptotic upper and lower bounds between the current state of the system and the limiting measure, which allows to infer window cutoff convergence for Brownian perturbations. The proof techniques in Section~\ref{s:Gaussianoblanco} are not specific to Brownian perturbations, and we show in Section~\ref{s:Gaussianorojo}, how to adapt them adequately by an optimal replica (coupling) to stationary Brownian Ornstein-Uhlenbeck perturbations. Section~\ref{s:ruidoestable} shows how our findings in the previous sections extend to $\alpha$-stable drivers, when we leave the Gaussian paradigm. 
In Appendix~\ref{appendix:WS} we show the non-standard ``shift linearity'' property of the WKR distance and calculate the moments of the $\alpha$-stable limiting laws. 

\section{\textbf{The setup and basic notation}}\label{s:setup} 

\noindent  
For a given sequence $x=(x_n)_{n\in \mathbb{N}}\in \RR^{\NN}$ and $\nu>0$,
let us consider the following infinite system of equations
\begin{equation}\label{e:sistema}
\begin{split}
A_1(t;x) &= x_1 + \int_0^t (-A_2(s;x) - \nu A_1(s;x)) \ud s + L(t),\\
A_n(t;x) &= x_n + \int_0^t ({A_{n-1}(s;x)}-A_{n+1}(s;x) - \nu A_n(s;x)) \ud s, \quad n\gqq 2,\, t\gqq 0,
\end{split}
\end{equation}
where $(L(t))_{t\gqq 0}$ is a stochastic process. 
In the sequel, $L$ will be a Brownian motion, a L\'evy process or an Ornstein-Uhlenbeck process. 
The system~\eqref{e:sistema} is an infinitely dimensional non-homogeneous linear system and its unique (path-wise) solution 
$A(t;x) = (A_n(t;x))_{n\in \NN}$ starting in $x$ 
can be decomposed, by the variation of constants formula, as
\begin{align}
A_n(t;x) = d_n(t;x) + C_n(t),\quad n\in \NN,\, t\gqq 0,  
\end{align}
where $d(t;x) = (d_n(t;x))_{n\in \NN}$ is the deterministic solution  
of the homogeneous system starting in $x$ and $C(t) = (C_n(t))_{n\in \NN}$ 
is the in-homogeneous solution starting in $0$, that is,
$(d(t;x))_{t\gqq 0}$ is the solution of~\eqref{e:sistema} with  $L\equiv 0$ and initial condition $x\in \mathbb{R}^\mathbb{N}$, while
$(C(t))_{t\gqq 0}$ is the solution of~\eqref{e:sistema} with initial condition $0\in \mathbb{R}^\mathbb{N}$.
Note that the random term $C(t)$ does not depend on $x$. 
By Proposition~1 in~\cite{PS08} we have 
\begin{align}\label{e:det} 
d_n(t;x) &= e^{-\nu t} \sum_{m=1}^\infty x_m \Big(J_{|n-m|}(2t) + (-1)^{m-1}J_{n+m}(2t)\Big) \quad 
\textrm{ for }\quad n\gqq 1,\, t\gqq 0, 
\end{align}
and 
\begin{equation}
C_n(t)= \int_0^t H_n(t-r) \ud L(r), \quad n\gqq 1,\, t\gqq 0,
\end{equation}
where
\begin{equation}
H_n(r) = n \frac{J_n(2r)}{r} e^{-\nu r},\quad n\gqq 1,\, r>0,  
\end{equation}
and $J_n$ is the Bessel function of the first kind with index $n$, that is,
\begin{equation}
J_n(s) = \sum_{k=0}^\infty \frac{(-1)^k}{k! (k+n)!} \Big(\frac{s}{2}\Big)^{2k + n},\quad s\gqq 0. 
\end{equation}
We stress that the solution $(A(t;x))_{t\gqq 0}$ exists component-wise in the previous sense.
We consider the space of all square-summable real sequences
\[\ell_2:=\left\{(x_n)_{n\in \mathbb{N}}\in \mathbb{R}^\mathbb{N}: \sum_{n=1}^\infty x^2_n<\infty\right\}
\]
endowed with the canonical inner product $\langle \cdot, \cdot\rangle$
and denote by $\|\cdot \|$ its corresponding norm. We note that $(\ell_2,\|\cdot\|)$ is a Hilbert space.

\noindent
We point out that the deterministic solution is asymptotically exponentially stable in $\ell_2$. More precisely, 
for initial datum $x\in \ell_2$ and $t\gqq 0$ it follows that 
\begin{equation}
\|d(t; x)\|=e^{-\nu t} \|x\|,
\end{equation}
see Lemma~\ref{lem:Lyapunov} below.

\noindent
For the case of $L = B$ a scalar standard Brownian motion, it is shown for any fixed $n\in \NN$ and $t>0$ in~\cite{PS08} that each $A_n(t;0)$ is a centered Gaussian random variable with variance $\sigma^2_n(t)$.
It converges in law as $t$ tends to infinity to a centered Gaussian limiting law with variance $\sigma^2_n$.
Since $\nu>0$, Lemma~\ref{lem:sumavarian} in Appendix~\ref{ap:subOUsystem} yields that
$\sum_{n=1}^{\infty} \sigma^2_n(t)\lqq
\sum_{n=1}^{\infty} \sigma^2_n<
\infty$.

\noindent
For random elements $X_i$, $i=1,2$ with values in $\ell_2$ with $\EE[\| X_i\|^p]< \infty$, for some $p\gqq 1$, we define the classical Wasserstein-Kantorovich-Rubinstein (WKR) $p$-distance over $\ell_2$ between $X_i$, $i=1,2$ by 
\begin{equation}
\mathcal{W}_{p}(X_1, X_2) := \Big(\inf_{\pi \in \mathcal{C}(X_1, X_2)} \iint_{\ell_2\times \ell_2}\|u-v \|^p \pi(\ud u, \ud v) \Big)^\frac{1}{p},
\end{equation}
where $\cC(X_1, X_2)$ is the family of couplings (joint laws) between the laws of $X_i$, $i=1,2$, see 
also~\cite{Panaretos,Vi09}. 
Note that $\wW_{p}$ defines a complete metric space on the set of probability distributions on $\ell_2$ (equipped with its Borel-sigma algebra) with finite $p$-th moments.  

\noindent
We recall the following properties of $\wW_{p}$:  
\begin{enumerate}
 \item \textbf{Translation invariance:} for (deterministic) $u, v\in \ell_2$ and random elements $X_i$, $i = 1,2$ with values in $\ell_2$ and $\EE[\|X_i\|^p] <\infty$, $i=1,2$, for some $p\gqq 1$ we have 
 \[
 \wW_{p}(u + X_1, v+X_2) =   \wW_{p}(u -v+ X_1, X_2) = \wW_{p}(X_1, v-u+ X_2). 
 \]
 \item \textbf{Shift linearity:} for (deterministic) $u\in \ell_2$ and a random element $X$ with values in $\ell_2$ with $\EE[\| X\|^p]<\infty$ for some $p\gqq 1$ 
 we have 
 \begin{equation}\label{e:shift}
 \wW_{p}(u + X, X) =  \wW_{p}(X, u+ X) = \|u\|. 
 \end{equation}
\end{enumerate} 
Property~(1) is classical and can be found for instance 
in~\cite{Vi09}. Property~(2) is non-standard and has been shown first in~\cite[Lemma 2.2~(d)]{BHPWA}. For completeness, (2)~is shown in Appendix~\ref{appendix:WS}. 

\section{\textbf{{Abrupt thermalization for Brownian perturbation with fixed variance}}}\label{s:Gaussianoblanco}

\noindent
In this section, we show ergodic cutoff estimates of the 
system~\eqref{e:sistema} driven by a Brownian $L=\sigma B$ with a fixed variance $\sigma^2>0$. We stress that no asymptotically small prefactor in front of the noise is involved, in contrast 
to~\cite{BHPWA,BHPSPDE,BHPTV,BHPWANO} or 
typical Freidlin-Wentzell theory, see for instance~\cite{BEGK04, Br91, Debussche, FW98, IP06}, among others. However, the main result of this section (Theorem~\ref{thm:cutoffconvergence}) is applied in the small noise setting in Corollary~\ref{cor:small} below. 

\noindent
Let us denote by $\mathcal{N}(m, C)$ a normal distribution with mean $m$ and covariance operator $C$ on the respective space. We say, that a random vector $X$ satisfies $X \stackrel{d}{=} \mathcal{N}(m, C)$ if the law of $X$ is equal to $\mathcal{N}(m, C)$ on the respective space. 
By~\cite{MSV07, PS08} we have the explicit identity in law 
\begin{align}
A(t;x)  \stackrel{d}{=} \mathcal{N}(d(t;x), \Sigma_t)\quad \textrm{ on }\quad t\gqq 0,\, x\in \ell_2,  
\end{align}
where $d(t;x)$ is given by~\eqref{e:det} and there is a unique Gaussian invariant limiting distribution     
$\mathcal{G} \stackrel{d}{=} \mathcal{N}(0, \Sigma_\infty)$ 
in $\ell_2$, see Lemma~\ref{lem:sumavarian}, with the closed form covariance operators 
\begin{align}
\Sigma_t = \Big({\sigma^2}\int_0^t H_i(r)H_{j}(r) \ud r\Big)_{i,j\in \NN}\quad \textrm{ and }\quad \Sigma_\infty = \Big({\sigma^2}\int_0^\infty H_i(r)H_{j}(r) \ud r\Big)_{i,j\in \NN}.
\end{align}
The detailed computations of $\Sigma_\infty$ is given in~\cite[Section~4.2]{PS08}. 
For convenience and further use, we denote by $\mathcal{G}_n$ the projection of $\mathcal{G}$ to the $n$-th coordinate in $\ell_2$. 

\noindent
We now state the first main result of this article. 

\begin{theorem}[Ergodic WKR bounds for $L=\sigma B$]\label{thm:cutoffconvergence}
Set 
\[
t_\e := \frac{1}{\nu} \ln(1/\e), \quad \e\in (0,1]. 
\]
Then for any $x\in \ell_2$, $p\gqq 1$, $\e\in (0,1]$ and $r> -t_\e$ it follows that
\begin{align}\label{e:perfil0}
\quad e^{-\nu\cdot r} \|x\|\lqq \frac{\wW_{p}(A(t_\e + r;x), \gG)}{\e} \lqq e^{-\nu\cdot r} (\| x\| + \EE [\| \gG\|]),
\end{align}
\end{theorem}
\noindent
Note that the bounds in inequality~\eqref{e:perfil0} do not depend on $p\gqq 1$. We now apply the bounds in Theorem~\ref{thm:cutoffconvergence} and show  asymptotic cutoff stability in the sense of~\cite{BHChaos}.

\begin{corollary}[Window cutoff convergence for $L=\sigma B$]\label{cor:window} Assume the hypotheses of Theorem~\ref{thm:cutoffconvergence}.  
Then for any $x\in \ell_2$ and $p\gqq 1$ it follows that 
\begin{align}\label{e:window}
&\lim\limits_{r\ra-\infty}\liminf\limits_{\e \ra 0} \frac{\wW_{p}(A(t_\e + r;x), \gG)}{\e} = \infty,\\
&\lim\limits_{r\ra\infty}\limsup\limits_{\e \ra 0}  \frac{\wW_{p}(A(t_\e + r;x), \gG)}{\e} = 0.
\end{align}
\end{corollary}

\noindent
In the sequel, let $(A^\e(t;x))_{t\gqq 0}$ be the solution of~\eqref{e:sistema}, 
where instead of $L =\sigma B$ we consider $L = \e B$, for $\e\ra 0$. 
In other words,
\begin{equation}\label{e:sistemaRP}
\begin{split}
A_1^\e(t;x) &= x_1 + \int_0^t (-A_2^\e(s;x) - \nu A_1^\e(s;x)) \ud s + \e B(t),\\
A_n^\e(t;x) &= x_n + \int_0^t ({A_{n-1}^\e(s;x)}-A_{n+1}^\e(s;x) - \nu A_n^\e(s;x)) \ud s, \quad n\gqq 2,\, t\gqq 0.
\end{split}
\end{equation}
We denote the $\e$-dependent invariant measure by 
\begin{equation}\label{eq:defprime}
\gG^\e \stackrel{d}{=} \e \gG',\quad \textrm{ where } 
\gG'
\textrm{ is given in Theorem~\ref{thm:cutoffconvergence} for } \sigma=1.
\end{equation}
We point out that
$\EE[\|\gG^\e\|]\ra 0$ as $\e\ra 0$.
The previous results, Theorem~\ref{thm:cutoffconvergence} and Corollary~\ref{cor:window}, can be further sharpened to a profile cutoff thermalization as follows. 
\begin{corollary}[Profile cutoff thermalization for $L=\e B$, $\e\ra 0$]\label{cor:small}
Set 
\[
t_\e := \frac{1}{\nu} \ln(1/\e), \quad \e\in (0,1]. 
\]
Then for any $x\in \ell_2$, $p\gqq 1$, $\e\in (0,1]$ and $r>-t_\e$ we have 
\begin{align}\label{e:cotassmall}
\quad e^{-\nu\cdot r} \| x\|\lqq \frac{\wW_{p}(A^\e(t_\e + r;x), \gG^\e)}{\e} \lqq e^{-\nu\cdot r}  (\|x\| +  \EE [\|{\gG^\e}\|]
), 
\end{align}
where
$\mathcal{G}^\e$ satisfies~\eqref{eq:defprime}.
 Moreover, it follows that
\begin{align}\label{e:perfilsmall}
\lim_{\e\ra 0} \frac{\wW_{p}(A^\e(t_\e + r;x), \gG^\e)}{\e} = e^{-\nu\cdot r}\|x\|.
\end{align}
\end{corollary}
\noindent
In the sequel, we show Theorem~\ref{thm:cutoffconvergence} in four lemmas. 
\begin{lemma}[Upper and lower bounds for $\wW_{p}(A(t;x), \gG)$]\label{lem:cotas} We keep the notation of Theorem~\ref{thm:cutoffconvergence}. Then for any $x\in \ell_2$ and $t\gqq 0$ it follows that 
 \begin{align}
\|d(t;x)\| \lqq \wW_{p}(A(t;x), \gG) \lqq \|d(t;x)\|+\wW_{p}(C(t), \gG).
\end{align}
\end{lemma}

\begin{proof}
We start with the following estimate using the triangle inequality, 
translation invariance and the shift linearity 
\begin{align}
\wW_{p}(A(t;x), \gG) &= \wW_{p}(d(t;x) + C(t), \gG)\\
&\lqq \wW_{p}(d(t;x) + C(t), d(t;x) + \gG)+\wW_{p}(d(t;x) + \gG, \gG)\\
&= \wW_{p}(C(t), \gG)+ \| d(t;x)\|. 
\end{align}
We continue with the lower bound. Using that $\EE[\gG] = \EE[ C(t)]$ we have 
for any coupling $\pi$ of $A(t;x)$ and $\gG$ that 
\begin{align}
\|d(t;x)\| &= \|  \EE[d(t;x) + C(t)- \gG]\|
= \left\| \iint_{\ell_2\times \ell_2}   (u - v) \pi(\ud u, \ud v)\right\|\\
&\lqq  \iint_{\ell_2\times \ell_2} \|  u - v\| \pi(\ud u, \ud v). 
\end{align}
Optimizing over all couplings we obtain 
$\| d(t;x)\| \lqq \wW_{1}(A(t;x), \gG)$. 
Using Jensen's inequality, we have for any $p\gqq 1$ that 
\[
\|d(t;x)\| \lqq \wW_{p}(A(t;x), \gG).
\]
This finishes the proof. 
\end{proof}

\begin{remark}\label{rem}
Note that in Lemma~\ref{lem:cotas} we do not use any specific Gaussian structure of $\gG$ and $A(t;x)$. For the upper bound we only use the general properties of $\wW_{p}$: the triangular inequality, the translation invariance, the shift linearity as well as finiteness of $p$-th moments of the laws. 
For the lower bound we only use that $\EE[\mathcal{G}]=\EE[ C(t)]\in \ell_2$, $t\gqq 0$, and that the $p$-th Gaussian moments are finite.
\end{remark}

\noindent
We start with the analogue to Lemma~B.1 in~\cite{BJ1}.
\begin{lemma}[Lyapunov exponent]\label{lem:Lyapunov} 

We keep the notation of Theorem~\ref{thm:cutoffconvergence}.
For any $x\in \ell_2$ and $t\gqq 0$ 
it follows that 
\begin{align}\label{e:estable}
\| d(t; x)\|= e^{-\nu t}\| x\|. 
\end{align} 
\end{lemma}

\begin{proof}
Recall that
$(d(t;x))_{t\gqq 0}$ is the solution of~\eqref{e:sistema} with  $L\equiv 0$ and initial condition $x\in \mathbb{R}^\mathbb{N}$. In other words, it satisfies
\begin{equation}\label{eq:defsystema}
\frac{\ud}{\ud t}d(t;x)=(\cR-\nu \cI)d(t;x),\quad t\gqq 0 \quad \textrm{ with }\quad d(0;x)=x,
\end{equation}
where $\cI$ denotes the identity operator in $\ell_2$ and $\cR$ is the operator in $\ell_2$ given in the canonical basis by
$\cR:=(r_{j,k})_{j,k\in \mathbb{N}}$ with
\begin{equation}
r_{j,k}:=
\begin{cases}
-1 & \textrm{ if } \quad k=j+1, j\in \mathbb{N},\\
1 & \textrm{ if } \quad k=j-1, j\in \mathbb{N}\setminus\{1\},\\
0 &  \textrm{ otherwise}.
\end{cases}
\end{equation}
Since~\eqref{eq:defsystema} is a linear homogeneous differential equation (infinite dimensional), we have
\begin{equation}
d(t;x)=\exp((\cR-\nu \cI)t)x\quad \textrm{ for all }\quad t\gqq 0.   
\end{equation}
We note that the operators $\cI$ and $\cR$ commute. By the celebrated Baker-Campbell-Hausdorff-Dynkin formula, see~\cite[ Chapter 5]{Hall}, we have the exponential representation
\begin{equation}
d(t;x)=\exp(-\nu \cI t)\exp(\cR t)x=e^{-\nu t}\exp(\cR t)x\quad \textrm{ for all }\quad t\gqq 0,   
\end{equation}
which yields 
$\|d(t;x)\|=e^{-\nu t}\|\exp(\cR t)x\|$ for all $t\gqq 0$.

\noindent
In the sequel, we calculate $\|\exp(\cR t)x\|$.
Let $\ell_2(\mathbb{C}):=\{z=(z_n)_{n\in\mathbb{N}}\in \mathbb{C}^{\mathbb{N}}\colon\sum_{n=1}^\infty |z_n|^2<\infty\}$. 
The canonical norm of $\ell_2(\mathbb{C})$ is denoted
$\|\cdot\|_{\ell_2(\mathbb{C})}$.
Let $\mathsf{i}$ be the unit imaginary and set $\cS:=\mathrm{i}\cR$. We observe that $\cS^*=\cS$, that is, $\cS$ is a Hermitian operator in $\ell_2(\mathbb{C})$.
Since
\[
\exp(\cR t)=\exp(-\mathsf{i}\cS t),\quad t\gqq 0,
\]
Stone's Theorem \cite[Theorem~A]{Stone} yields that $\exp(-\mathsf{i}\cS t)$ is unitary in $\ell_2(\mathbb{C})$, 
which implies 
\begin{equation}
\begin{split}
\|\exp(\cR t)x\|_{\ell_2(\mathbb{C})}=
\|x\|_{\ell_2(\mathbb{C})}=\|x\|\quad \textrm{ for all }\quad t\gqq 0. 
\end{split}
\end{equation}
This completes the proof. 
\end{proof}

\begin{lemma}[Disintegration]\label{lem:desintegracion} 
We keep the notation of Theorem~\ref{thm:cutoffconvergence}.
For any $t\gqq 0$ it follows 
for all $p\gqq 1$ that 
\begin{align}
\wW_{p}(C(t), \gG)\lqq e^{-\nu t} \EE[\| \gG\|],
\end{align}
where $\EE[\| \gG\|] = \int_{\ell_2} \| y\| \PP( \gG \in \ud y)<\infty$. 
\end{lemma}

\noindent
Note that the right-hand side does not depend on the parameter $p$ and it is finite by Lemma~\ref{lem:sumavarian} in Appendix~\ref{ap:subOUsystem}.
\begin{proof} Since $A(t; \gG) \stackrel{d}{=} \gG$ for any $t\gqq 0$, and with the help of
Lemma~\ref{lem:Lyapunov} we have 
\begin{align}
\wW_{p}(C(t), \gG) &=\wW_{p}(C(t), A(t;\gG))\lqq \int_{\ell_2} \wW_{p}(C(t), A(t;y)) \PP( \gG \in \ud y)\\ 
&=\int_{\ell_2} \| d(t;y)\|  \PP( \gG \in \ud y) 
= e^{-\nu t}\int_{\ell_2} \| y\| \PP( \gG \in \ud y), 
\end{align}
where we use disintegration with the help of the Markov property 
of~\eqref{e:sistema} and the shift 
linearity~\eqref{e:shift}. 
\end{proof}

\begin{lemma}[Exponential ergodicity with respect to $\wW_{p}$]\label{lem:ergoblanco}

We keep the notation of Theorem~\ref{thm:cutoffconvergence}.
For $x\in \ell_2$ there exists a unique limit measure $\gG$ of the solution of the system~\eqref{e:sistema} and it follows that  
\begin{equation}\label{e:exergo}
\wW_{p}(A(t;x), \gG) \lqq e^{-\nu t}(\| x\| + \EE[\|\gG\|]), \quad t\gqq 0. 
\end{equation}
\end{lemma}

\begin{proof} 
The proof of~\eqref{e:exergo} follows by  Lemma~\ref{lem:cotas}, Lemma~\ref{lem:Lyapunov} and Lemma~\ref{lem:desintegracion}, 
as long as $\EE[\|\gG\|]<\infty$. 
Indeed, using Jensen's inequality and monotone convergence we calculate 
\begin{align}\label{e:cotaprimeroporsegundo}
\EE[\|\gG\|] = \EE\Big[\sqrt{\| \gG\|^2}\Big] = \EE\bigg[\Big(\sum_{n=1}^\infty  \gG_n^2\Big)^{1/2}\bigg] 
\lqq 
\bigg(\EE\bigg[\sum_{n=1}^\infty  \gG_n^2\bigg]\bigg)^{1/2}<\infty,
\end{align}
due to Lemma~\ref{lem:sumavarian} in Appendix~\ref{ap:subOUsystem}. 
\end{proof}

\begin{remark}
We recall that for $p=2$, $N_i \stackrel{d}{=}\mathcal{N}(0, \Sigma_i)$ with values in $\ell_2$, $i=1,2$, 
there is the explicit Gaussian formula (see~\cite{Gel})
\begin{equation}\label{e:Gelbrich}
\wW_{p}(N_1, N_2) = \mathsf{Tr}(\Sigma_1) 
+\mathsf{Tr}(\Sigma_2) -\mathsf{Tr}((\Sigma_1^{\frac{1}{2}} \Sigma_2  \Sigma_1^{\frac{1}{2}})^\frac{1}{2}), 
\end{equation}
where $\mathsf{Tr}$ denotes the trace operator. 
The preceding formula should lead to better bounds that the right-hand side of~\eqref{e:exergo} of Lemma~\ref{lem:ergoblanco}. However, the right-hand side is hard to assess due to the concatenation of operator squares root in the trace operators, which should cancel in order to obtain ergodicity, $\wW_{p}(A(t;x), \gG)\ra 0$ as $t\ra \infty$. The trade-off is that the ``error'' term $\EE[\| \gG\|]$ comes with the identical exponential rate, instead of a faster decay. 

\noindent
In our case, however, Lemma~\ref{lem:ergoblanco} does not depend on Gaussianity, and is robust for other drivers. 
It remains valid for any $p\gqq 1$ as long as the respective laws $\mu_i$, $i=1,2$, are supported on $\ell_2$ and satisfy 
\[
\int_{\ell_2} \| z\|^p \mu_i(\ud z) < \infty.  
\]
\end{remark}

\begin{proof}[Proof of Theorem~\ref{thm:cutoffconvergence}:]
We start with the upper bound of~\eqref{e:perfil0}. By Lemma~\ref{lem:cotas}, the upper bound in~\eqref{e:estable} of 
Lemma~\ref{lem:Lyapunov} and 
Lemma~\ref{lem:desintegracion} 
we have 
for any $\e\in (0,1]$, $x\in \ell_2$ and $t\gqq 0$ that  
\begin{align}\label{e:crudeupper}
\frac{\wW_{p}(A(t;x), \gG)}{\e}  \lqq \frac{e^{-\nu t}}{\e} (\| x\| + \EE [\| \gG\|]).  
\end{align}
In particular, $t = t_\e + r$ yields 
\begin{align}
\frac{\wW_{p}(A(t_\e + r;x), \gG)}{\e}  \lqq e^{-\nu\cdot r} (\|x\| + \EE [\|\gG\|]).  
\end{align}
We continue with the lower bound. 
By Lemma~\ref{lem:cotas} and the lower bound in~\eqref{e:estable} of Lemma~\ref{lem:Lyapunov} 
we have for any $\e\in (0,1]$, $x\in \ell_2$ and $t\gqq 0$ that 
\begin{align}\label{e:crudelower}
\frac{\wW_{p}(A(t;x), \gG)}{\e} \gqq \frac{\|d(t;x)\|}{\e} 
&= \frac{e^{-\nu t}}{\e} \| x\|. 
\end{align}
Evaluating $t = t_\e + r$ yields 
\begin{align}
\frac{\wW_{p}(A(t_\e + r;x), \gG)}{\e}  
&\gqq e^{-\nu\cdot r}  \|x\|. 
\end{align}
This finishes the proof. 
\end{proof}

\section{\textbf{Abrupt thermalization for Ornstein-Uhlenbeck perturbation with fixed variance}}\label{s:Gaussianorojo}

\noindent We study the system~\eqref{e:sistema} with $L(t) = U(t)$, $(U(t))_{t\gqq 0}$ being an Ornstein-Uhlenbeck process satisfying 
\begin{equation}\label{e: OU}
\ud U(t)  = -\gamma U(t) \ud t + \sigma \ud B(t), \quad U(0) \stackrel{d}{=} U_0, \quad \sigma>0,\, \gamma>0,\, x_0\in \RR,
\end{equation}
where $B = (B(t))_{t\gqq 0}$ is a scalar standard Brownian motion, and $U_0 \stackrel{d}{=} \mathcal{N}(0, \frac{\sigma^2}{2\gamma})$. Note that $U_0 \stackrel{d}{=} U(t;U_0)$ for all $t\gqq 0$ and $U_0$ is chosen independent of $(B(t))_{t\gqq 0}$. In order to retain the Markov property we consider the extended system, where $U(t) = A_0(t)$, 
\begin{equation}\label{e:aumentado}
\begin{split}
A_0(t;x)  &= U_0+\int_0^t (-\gamma A_0(s;x_0)) \ud s + \sigma B(t),\\
A_1(t;x) &= x_1 + \int_0^t (-A_2(s;x) - \nu A_1(s;x)) \ud s  + A_0(t),\\
A_n(t;x) &= x_n + \int_0^t ({A_{n-1}(s;x)}-A_{n+1}(s;x) - \nu A_n(s;x)) \ud s, \quad n\gqq 2,\, t\gqq 0.
\end{split}
\end{equation}
It is clear that~\eqref{e:aumentado} defines a Markovian process $A(t;x) = (A_n(t;x))_{n\in \NN_0}$. 
For convenience, we often write $A(t;x) = (A_0(t;x), A_+(t;x))$.  
The extended system~\eqref{e:aumentado} lives naturally in the state space $\RR\times \ell_2$, while~\eqref{e:sistema} has values in $\ell_2$ with the norm $\|\cdot\|$.  
Therefore, we naturally extend the notation from $\ell_2$ to $\RR\times \ell_2$. 
All properties remain valid. 

\noindent
In the sequel, we extend the space to the new state space $\RR\times \ell_2$ with the metric $\|(x_0, x)\|_0 := |x_0| + \|x\|, x\in \ell_2$.  We denote the WKR distance on $\RR\times \ell_2$ by 
$\wW_{0,p}$, and maintain all the previous notation, with the enhancement by the zero-th component, mutatis mutandis. 

\noindent
It is not hard to see that the extended system has a unique invariant Gaussian probability distribution $\widetilde \gG \stackrel{d}{=} \mathcal{N}(0, \widetilde \Sigma_\infty)$ with values in $\RR\times \ell_2$ equipped with $\|\cdot \|_0$, in other words, $\widetilde \gG \stackrel{d}{=} A(t; \widetilde \gG)$ for all $t\gqq 0$. 
Note that the zero-th component $A_0$ does not depend functionally on $A_+$, hence $A_0(t;\widetilde \gG) = A_0(t; \widetilde \gG_0)$, where $\widetilde \gG_0$ is the projection of $\widetilde \gG$ to the zero-th component.  
In Appendix~\ref{ss:OUCov} it is shown that
\[
\EE[|\gG_n|^2] \lqq 4\gamma \int_0^\infty H_n^2(r) \ud r,
\]
which implies $\EE[\|\widetilde \gG\|_0^2] <\infty$. 
Consequently, Lemma~\ref{lem:ergoblanco} remains valid and 
\begin{equation}\label{e:Lyapunovrojo}
\wW_{0,p}(A(t;(\widetilde \gG_0, x)), \widetilde \gG) \lqq e^{-\nu t} (\|x\|_0 + \EE[\| \widetilde \gG\|_0]). 
\end{equation}
Combining Lemma~\ref{lem:cotas}, Lemma~\ref{lem:Lyapunov} 
and~\eqref{e:Lyapunovrojo} we obtain the following.

\begin{theorem}[Ergodic WKR bounds for $L=U$]\label{thm:convergenciaroja}
Set 
\[
t_\e := \frac{1}{\nu} \ln(1/\e), \quad \e\in (0,1]. 
\]
Then for any $x\in \ell_2$, $p\gqq 1$, $\e\in (0,1]$ and $r>-t_\e$ it follows that
\begin{align}\label{e:perfilred}
e^{-\nu\cdot r}\|x\|
&\lqq \frac{\wW_{0,p}(A(t_\e + r; (\widetilde \gG_0, x)), \widetilde \gG)}{\e}\lqq 
e^{-\nu \cdot r} (\|x\| + \EE[\| \widetilde \gG_+\|]).
\end{align} 
\end{theorem}

\noindent
Note that the inequality~\eqref{e:perfilred} is valid for any $p\gqq 1$. 

\begin{proof}
We start with the upper bound. Fix $x \in \ell_2$. Then by the Markov property, disintegration and the shift linearity we have 
\begin{align}
&\wW_{0,p}\bigg(A(t;(\widetilde \gG_0,x)), \widetilde \gG\bigg) = \wW_{0,p}\bigg(A(t;(U_0,x)), \widetilde \gG\bigg)\\ 
&=\wW_{0,p}\bigg(\Big(\begin{array}{l}A_0(t;U_0) \\ A_+(t;(U_0,x))\end{array}\Big), \widetilde \gG\bigg)
= \wW_{0,p}\bigg(\Big(\begin{array}{l}A_0(t;U_0) \\ A_+(t;(U_0,x))\end{array}\Big),
\Big(\begin{array}{l}A_0(t;\widetilde \gG) \\ A_+(t;\widetilde \gG)\end{array}\Big)\bigg)\\
&= \wW_{0,p}\bigg(\Big(\begin{array}{l}A_0(t;U_0) \\ A_+(t;(U_0,x))\end{array}\Big),
\Big(\begin{array}{l}A_0(t;\widetilde \gG_0) \\ A_+(t;\widetilde \gG)\end{array}\Big)\bigg)\\
&\lqq \int_{\RR} \int_{\RR\times \ell_2} 
\wW_{0,p} \bigg(\Big(\begin{array}{l}A_0(t;u) \\ A_+(t;(u,x))\end{array}\Big), 
\Big(\begin{array}{l}A_0(t;v) \\ A_+(t; (v,y))\end{array}\Big)\bigg) \pi(U_0 \in \ud u, (\widetilde \gG_0, \widetilde \gG_+) \in ( \ud v, \ud y))\\
&= \int_{\RR} \int_{\RR\times \ell_2} 
\Big(\|d(t;x-y)\|+ e^{-\gamma t} |u-v|\Big)
 \pi(U_0 \in \ud u, \widetilde \gG \in (\ud v, \ud y))
\end{align}
for any coupling $\pi$ between $U_0$ and $(\widetilde \gG_0,\widetilde \gG_+)$. In particular, for any coupling between the synchronous coupling $U_0 = \widetilde \gG_0$ and $\widetilde \gG_+$.  
Hence, 
\begin{align}
\wW_{0,p}(A(t;(\widetilde \gG_0,x)), \widetilde \gG) 
&\lqq e^{-\nu t} \EE[\|x-\widetilde \gG_+\|]. 
\end{align}
We continue with the lower bound. Note that in total generality we have for random vectors $(U_0, U_+)$ and $(G_0, G_+)$ with $\EE[G_+]=0$ 
\begin{align}
&\wW_{p}\bigg(\Big(\begin{array}{c} U_0 \\ U_+\end{array}\Big), \Big(\begin{array}{c} G_0 \\ G_+\end{array}\Big)\bigg) 
\gqq \wW_{1}\bigg(\Big(\begin{array}{c} U_0 \\ U_+\end{array}\Big), \Big(\begin{array}{c} G_0 \\ G_+\end{array}\Big)\bigg)\\
&= \inf_{\pi \in \cC((U_0, U_+), (G_0, G_+))} \iint_{\RR\times \ell_2} \Big(|u_0-g_0|+\|u_+ -g_+\|\Big) \pi\Big(\begin{array}{c} (U_0, U_+) \in (\ud u_0, \ud u_+)\\ (G_0, G_+) \in (\ud g_0, \ud g_+)\end{array}\Big) \\
&\gqq \inf_{\pi \in \cC((U_0, U_+), (G_0, G_+))}\Big| \iint_{\RR\times \ell_2} \big(u_+ - g_+\big) 
\pi\Big(\begin{array}{c} (U_0, U_+) \in (\ud u_0, \ud u_+)\\ (G_0, G_+) \in (\ud g_0, \ud g_+)\end{array}\Big)\Big| \\
&= \Big| \EE[U_+] - \EE[G_+] \Big| = | \EE[U_+] |. 
\end{align}
For $(G_0, G_+) = (\widetilde \gG_0, \widetilde \gG_+)$ and $U_0 = A_0(t;\widetilde \gG_0)$ and $U_+ = A_+(t;(\widetilde \gG_0, x))$ we infer due to $\EE[\widetilde \gG_0] = 0$ the estimate  
\begin{align}
\wW_{0,p}(A(t;(\widetilde \gG_0,x)), \widetilde \gG)
&\gqq \|\EE[A_+(t;(\widetilde \gG_0, x))]\| = |e^{-\gamma t} \EE[\gG_0]| + \|d(t;x)\| \\
&= \|d(t;x)\| = e^{-\nu t}\|x\|. 
\end{align}
This finishes the proof. 
\end{proof}

\begin{corollary}[Window cutoff convergence for $L=U$]\label{cor:windowOUestable} Assume the hypotheses of 
Theorem~\ref{thm:cutoffconvergence}. 
Then for any $x\in \ell_2$ and $p\gqq 1$ it follows that
\begin{align}\label{e:windowred}
&\lim\limits_{r\ra-\infty}\liminf\limits_{\e \ra 0} \frac{\wW_{0,p}(A(t_\e + r; (\widetilde \gG_0, x)), \widetilde\gG)}{\e} = \infty,\\
&\lim\limits_{r\ra\infty}\limsup\limits_{\e \ra 0} \frac{\wW_{0,p}(A(t_\e + r; (\widetilde \gG_0, x)), \widetilde\gG)}{\e} = 0.
\end{align}
\end{corollary}

\begin{remark}\label{rem:smallGOU}
The analogous results for $\e$-small Ornstein-Uhlenbeck perturbations as in Corollary~\ref{cor:small} can be obtained similarly. 
\end{remark}

\section{\textbf{Abrupt thermalization for different types of L\'evy perturbations}}\label{s:ruidoestable}

\noindent It is well-known that Brownian motion is a particular example of the larger class of random drivers, namely the class of L\'evy processes. Recall that a L\'evy process is a c\`adl\`ag random process with stationary and independent increments starting in $0$. For details we refer to~\cite{APPLEBAUMBOOK, KP22,Sa}. 

\subsection{\textbf{$\alpha$-stable perturbations with fixed amplitude}}\label{ss:ruidoestable}\hfill\\

\noindent
In this subsection, we restrict our attention to the case of 
a symmetric $\alpha$-stable driver $L_\alpha:=(L_\alpha(t))_{t\gqq 0}$ for some $1 < \alpha < 2$ with characteristic exponent $\psi(u) = -\sigma^\alpha |u|^\alpha$, $u\in \RR$ for some fixed $\sigma>0$. 
It is shown in~\cite{PS08} that the solution of~\eqref{e:sistema} has the same shape when $\sigma B$ is replaced by $L=L_\alpha$.
In abuse of notation, we keep the analogous notation of the Gaussian system in Section~\ref{s:Gaussianoblanco}. 

\noindent
In the sequel, we verify that in this setting $\EE[\|A(t;x)\|]<\infty$ for any $t\gqq 0$ and $\EE[\| \gG\|]<\infty$ 
for the limit law $\gG$. 
We show that for any $x\in \ell_2$ and $t\gqq 0$ 
\[
\sum_{n=1}^\infty |A_n(t;x)|^2 <\infty \quad \textrm{ a.s.}   
\]
We start with the elementary observation that for any 
$1\lqq \eta < 2$ we have 
$\ell_\eta\subset \ell_2$,
where
\begin{equation}
\ell_\eta := \bigg\{(x_n)_{n\in \NN}\colon x_n\in \RR,\,\, \sum_{n=1}^\infty |x_n|^\eta <\infty\bigg\}. 
\end{equation}
Hence, it is sufficient to show that $\sum_{n=1}^\infty \EE[|A_n(t;x)|^\eta]<\infty$ for  $1\lqq \eta < \alpha <2$.  
Since $A(t;x) = d(t;x) + C(t)$ and since $x\in \ell_2$ implies $d(t;x)\in \ell_2$ for all $t\gqq 0$, 
it is sufficient to show that $\sum_{n=1}^\infty \EE[|C_n(t)|^\eta]<\infty$ for all $t\gqq 0$. 
For any $t\gqq 0$ we have 
\[
\EE\Big[e^{\mathsf{i} u C_n(t)}\Big] = \exp\Big(-\sigma^\alpha u^\alpha \int_0^t |H_n(s)|^\alpha  \ud s\Big), 
 \]
and sending $t\ra \infty$ we obtain by~(29) in~\cite{PS08} 
\[
\EE\Big[e^{\mathsf{i} u \gG_n}\Big] = \exp\Big(-\sigma^\alpha u^\alpha \int_0^\infty |H_n(s)|^\alpha  \ud s\Big). 
\]  
By~\cite{Sa} formula~(25.6) (or~\cite{KP22} Theorem~1.13) we have 
for a symmetric $(\alpha, \sigma)$-stable distribution $X$ with $\EE[e^{\mathsf{i} u X}] = e^{-\sigma^\alpha|u|^\alpha}$ 
the absolute moment of order $0 < \theta < \alpha$ 
\begin{align}
\EE[|X|^\theta] = \sigma^\theta 2^\theta  \frac{\Gamma(\frac{1+\theta}{2}) \Gamma(1-\frac{\theta}{\alpha})}{\sqrt{\pi} \Gamma(1-\frac{\theta}{2})}, 
\end{align}
where $\Gamma$ denotes the usual Gamma function. 
Hence, 
\begin{align}\label{e:momentotheta}
\EE[|\gG_n|^\theta] = 2^\theta \sigma^\theta \Big(\int_0^\infty |H_n(r)|^\alpha  \ud r\Big)^\frac{\theta}{\alpha} 
\frac{\Gamma(\frac{1+\theta}{2}) \Gamma(1-\frac{\theta}{\alpha})}{\sqrt{\pi} \Gamma(1-\frac{\theta}{2})}.
\end{align}

\noindent
Note that on the right-hand side $\theta = \alpha$ produces the factor $\Gamma(0) = \infty$. 
Hence, it is sufficient to show that for some $1 \lqq \theta < \alpha < 2$ 
\[
\sum_{n=1}^\infty \Big(\int_0^\infty |H_n(r)|^\alpha  \ud r\Big)^\frac{\theta}{\alpha} <\infty. 
\]
Observe that $H_n(r) = n \frac{J_n(2r)}{r} e^{-\nu r}$. Therefore, the preceding condition reads as follows 
\begin{equation}\label{e:Lambdaestable}
\sum_{n=1}^\infty n^\theta \Big(\int_0^\infty \frac{|J_n(2r)|^\alpha e^{-\alpha \nu r}}{r^\alpha} \ud r\Big)^\frac{\theta}{\alpha} <\infty.   
\end{equation}
This result is established in Lemma~\ref{lem:stablesummability}. 

\noindent
For the main result we use the shift linearity of $\wW_{p}$ for $p\gqq 1$ codified in Lemma~2.2 in~\cite{BHPWA}, which turns out to be false in general for $p<1$ (see~\cite[Remark~2.4]{BHPWA}). 

\begin{theorem}[Ergodic WKR bounds for $L=L_\alpha$]\label{thm:convergenciaestable}
Fix $1< \alpha < 2$. 
 Then for any $x\in\ell_2$, $1 \lqq p < \alpha$ and  
\begin{equation}\label{e:tiempo}
t_\e := \frac{1}{\nu} \ln(1/\e), \quad \e\in (0,1] 
\end{equation}
it follows that for all ${\e\in (0,1]}$ and $r>-t_\e$   
\begin{align}\label{e:perfilrosa}
\quad e^{-\nu\cdot r} \|x\|
&\lqq \frac{\wW_{p}(A(t_\e + r; x), \gG)}{\e}\lqq e^{-\nu \cdot r} (\|x\| + \EE[\|\gG\|]).
\end{align}
\end{theorem}

\noindent
The proof is a combination of Lemma~\ref{lem:cotas}, Lemma~\ref{lem:Lyapunov} and Lemma~\ref{lem:ergoblanco}. 
Note that the first two lemmas only depend on the existence of first order moments. 
Lemma~\ref{lem:ergoblanco} also remains valid, if we replace 
the second order moments $\EE[|\gG_n|^2]$ in 
formula~\eqref{e:cotaprimeroporsegundo} by $\EE[|\gG_n|^\theta]$ obtained in~\eqref{e:momentotheta} and apply 
condition~\eqref{e:Lambdaestable}. 

\noindent
We infer analogously cutoff convergence. 

\begin{corollary}[Window cutoff convergence for $L=L_\alpha$]\label{cor:windowestable} Assume the hypotheses of Theorem~\ref{thm:convergenciaestable}.  
Then for any $x\in \ell_2$ and $1\lqq p< \alpha$ it follows that 
\begin{equation}\label{e:windowestable}
\begin{split} 
&\lim\limits_{r\ra-\infty}\liminf\limits_{\e \ra 0} \frac{\wW_{p}(A(t_\e + r;x), \gG)}{\e} = \infty,\\
&\lim\limits_{r\ra\infty}\limsup\limits_{\e \ra 0}  \frac{\wW_{p}(A(t_\e + r;x), \gG)}{\e} = 0.
\end{split}
\end{equation}
\end{corollary}
\noindent
Small noise results similar to Corollary~\ref{cor:small} are obtained straightforwardly. 

\subsection{\textbf{$\alpha$-stable Ornstein-Uhlenbeck perturbations with fixed amplitude}} \hfill\\

\noindent We now study the system~\eqref{e:sistema} with $L(t) = U(t)$, $(U(t))_{t\gqq 0}$ being an $\alpha$-stable Ornstein-Uhlenbeck process satisfying 
\begin{equation}\label{e: OUestable}
\ud U(t)  = -\gamma U(t) \ud t + \sigma \ud L(t), \quad U(0) \stackrel{d}{=} U_0, \quad \sigma>0,\, \gamma>0,\, x_0\in \RR,
\end{equation}
where $L = (L(t))_{t\gqq 0}$ is a scalar symmetric $\alpha$-stable process with $1<\alpha<2$ 
\[
\EE[e^{\mathsf{i} u L(t)}] = e^{-t \sigma^\alpha |u|^\alpha} \quad\textrm{ for all }\quad u\in \RR,\, t\gqq 0.  
\]
The random initial data $U_0$ is distributed according to the 
invariant distribution of~\eqref{e: OUestable} and it has the characteristic function 
\begin{align}
\EE[e^{\mathsf{i} u U_0}] = e^{-|u|^\alpha \sigma^\alpha \int_0^\infty e^{-\gamma s \alpha} ds} = e^{-|u|^\alpha \frac{\sigma^\alpha}{\alpha \gamma}} \quad \textrm{ for all }\quad u\in \RR. 
\end{align}
For details see~\cite[Theorem 17.5]{Sa} and 
formula~\eqref{e:OUestableinvariante} in Appendix~\ref{a:stableOU}. 

\noindent
Note that $U_0 = U(t;U_0)$ in law for all $t\gqq 0$ and $U_0$ being independent from $(L(t))_{t\gqq 0}$. In the spirit 
of~\eqref{e:aumentado} we consider the extended system, where $U(t) = A_0(t)$, 
\begin{equation}\label{e:aumentadoOUestable}
\begin{split}
A_0(t;x)  &= U_0+\int_0^t (-\gamma A_0(s;x_0)) \ud s + L(t),\\
A_1(t;x) &= x_1 + \int_0^t (-A_2(s;x) - \nu A_1(s;x)) \ud s  + A_0(t),\\
A_n(t;x) &= x_n + \int_0^t ({A_{n-1}(s;x)}-A_{n+1}(s;x) - \nu A_n(s;x)) \ud s, \quad n\gqq 2,\, t\gqq 0.
\end{split}
\end{equation}
Again, we obtain that~\eqref{e:aumentadoOUestable} defines a Markovian process $A(t;x) = (A_n(t;x))_{n\in \NN_0}$ and maintain the notation $A(t;x) = (A_0(t;x), A_+(t;x))$.  
The extended system~\eqref{e:aumentadoOUestable} lives naturally in the state space $\RR\times \ell_2$, while~\eqref{e:sistema} has values in $\ell_2$ with the norm $\|\cdot \|$.
Therefore, analogously to Section~\ref{s:Gaussianorojo} we naturally extend the notation from $\ell_2$ to $\RR\times \ell_2$. 
All properties remain valid. 

\noindent
In particular, similarly to Lemma~\ref{lem:desintegracion} 
it is shown there, that whenever 
$\EE[\|\widetilde \gG\|_0]<\infty$, we have  
\[
\wW_{0,p}(A(t;x), \widetilde \gG) \ra 0\quad \textrm{ as }\quad t\ra\infty. 
\]

\begin{theorem}[Ergodic WKR bounds for $L=U$]\label{thm:convergenciarojaestable}
 Fix $1< \alpha < 2$  
 and
 \[
t_\e := \frac{1}{\nu} \ln(1/\e), \quad {\e\in (0,1]}. 
\]
Then for any $x\in \ell_2$, $1\lqq p< \alpha$, $\e\in (0,1]$ and $r>-t_\e$ 
it follows that  
\begin{align}\label{e:perfil}
e^{-\nu\cdot r} \|x\|
&\lqq \frac{\wW_{0,p}(A(t_\e + r; (\widetilde \gG_0, x)), \widetilde\gG)}{\e}\lqq e^{-\nu \cdot r} (\|x\| + \EE[\|\widetilde \gG_+\|]).
\end{align}
\end{theorem}

\begin{corollary}[Window cutoff convergence for $L=U$]\label{cor:windowOU} Assume the hypotheses of Theorem~\ref{thm:convergenciarojaestable}. 
Then for any $x\in \ell_2$ and $1\lqq p< \infty$ it follows that
\begin{align}\label{e:windowstablenoise}
&\lim\limits_{r\ra-\infty}\liminf\limits_{\e \ra 0} \frac{\wW_{0,p}(A(t_\e + r; (\widetilde \gG_0, x)), \widetilde\gG)}{\e} = \infty,\\
&\lim\limits_{r\ra\infty}\limsup\limits_{\e \ra 0} \frac{\wW_{0,p}(A(t_\e + r; (\widetilde \gG_0, x)), \widetilde\gG)}{\e} = 0.
\end{align}
\end{corollary}

\noindent
Small noise results similarly to Corollary~\ref{cor:small} can be obtained straightforwardly. 

\subsection{\textbf{The case of general L\'evy processes with second moments}}\hfill\\ 

\noindent For any centered L\'evy process $(L(t))_{t\gqq 0}$ with finite second moment, the characteristic function is 
given by 
\[
\RR \ni u \mapsto \EE[e^{\mathsf{i} u L(t)}] = e^{-t\Psi(u)}, \quad \textrm{ where }\quad \Psi(u) = \frac{\sigma^2 u^2}{2} + \int_{\RR} (e^{\mathsf{i}uy} - 1 - \mathsf{i}uy) \rho(\ud y), 
\]
where $\rho$ is the jump measure satisfying 
$\rho(\{0\}) = 0$ and $\int_{\RR} y^2 \rho(\ud y) < \infty$. Important examples are standard Brownian motion $\rho = 0$ treated in Section~\ref{s:Gaussianoblanco}, 
symmetric square-integrable compound Poisson processes, tempered $\alpha$-stable processes and two-sided symmetric $\Gamma$-processes. Note that $\alpha$-stable processes do not exhibit finite second moments. 

\begin{theorem}\label{thm:convergenciasegundo}
Consider the solution $(A(t;x))_{t\gqq 0}$ of system~\eqref{e:sistema} for initial data $x\in \ell_2$, where $L = (L(t))_{t\gqq 0}$ is a centered L\'evy process with $\EE[|L(1)|^2] = \frac{1}{2}$ and $\nu>0$. 
Define the time scale $(t_\e)_{\e\in (0,1]}$ by~\eqref{e:tiempo}. Then for any $x\in \ell_2$, $1\lqq p \lqq 2$, $r> -t_\e$ and $\e\in (0,1]$ 
the estimate~\eqref{e:perfilrosa} is valid. 
\end{theorem}

\begin{remark}~
\begin{enumerate}
\item Note that the lower bound in~\eqref{e:perfilrosa} is shown by Lemma~\ref{lem:cotas} and only depends on the first moments. 

\item In Lemma~\ref{lem:cotas} the upper bound is reduced to the ergodic bound treated in Lemma~\ref{lem:ergoblanco} 
with the help of the shift linearity for $p\gqq 1$, which only requires first moments. 
However, in order to avoid the technical difficulties in the calculus of the first absolute moment, 
the ergodic bound is dominated sub-optimally by the series second moments in~\eqref{e:cotaprimeroporsegundo}. 
While second moments can be obtained generically, the calculation of moments of lower order typically depends strongly on the underlying distribution. Hence, the condition of second moments can be removed case by case, 
as carried out in Subsection~\ref{ss:ruidoestable} for the $\alpha$-stable case $1<\alpha <2$. 

\item 
Due to the calculations~(28) in~\cite{PS08} note that
\begin{equation}\label{eq:defsegunda}
\EE[|\gG_n|^2] = -\Psi''(0) \int_0^\infty (H_n(r))^2 \ud r.  
\end{equation}
Since $\Psi''(0) = -\EE[L(1)^2]$ and $\EE[L(1)^2] = \frac{1}{2}$, Item~(2) 
with the help of~\eqref{eq:defsegunda} and Lemma~\ref{lem:sumavarian} gives Theorem~\ref{thm:convergenciasegundo} as in the Gaussian case.
\end{enumerate}
\end{remark}

\begin{corollary}[Window cutoff convergence for stable Ornstein-Uhlenbeck perturbations]\label{cor:windowsegundo} Assume the hypotheses of Theorem~\ref{thm:convergenciasegundo}. 
Then for any $x\in \ell_2$ and $1\lqq p \lqq 2$ we have the window cutoff~\eqref{e:windowestable}.
\end{corollary}

\noindent
The proof remains untouched. 

\section{\textbf{Appendix: Shift linearity and the characteristics of the limiting measures}} 
\subsection{\textbf{Proof of the Shift linearity (2)~for the WKR distance $\wW_{p}$ in $\ell_2$}}\label{appendix:WS}\hfill\\

\noindent Fix $p\gqq 1$. We first show the upper 
bound~\eqref{e:shift}. Consider the synchronous coupling $\pi$ between $X$ and $X$. 
Then by construction  
\begin{equation}\label{e:SLpordebajo}
\wW_{p}(u + X, u) \lqq \Big(\iint_{\ell_2\times \ell_2} \|(u + x)-x\|^p \pi(\ud x, \ud x)\Big)^{\frac{1}{p}} = \| u\|.
\end{equation}
For the lower bound of~\eqref{e:shift} we consider any coupling between $u + X$ and $X$. 
Then we have the following representation 
\begin{align}
\iint_{\ell_2\times \ell_2} (w - x) \pi(\ud w, \ud x) &= \iint_{\ell_2\times \ell_2} w \pi(\ud w, \ud x) - \iint_{\ell_2\times \ell_2} x \pi(\ud w, \ud x)\\
&
= \EE[u + X] - \EE[X] = u. 
\end{align}
Now, Jensen's inequality yields 
\begin{align}
\|u\| = \left\|\iint_{\ell_2\times \ell_2} (w - x) \pi(\ud w, \ud x)\right\| 
\lqq \iint_{\ell_2\times \ell_2} \|w - x \| \pi(\ud w, \ud x).
\end{align}
Minimizing over all possible couplings we obtain 
\begin{equation}\label{e:SLW1}
\| u\| \lqq \wW_{1}(u+X, X). 
\end{equation}
Finally,~\eqref{e:SLW1} and Jensen's inequality combined 
with~\eqref{e:SLpordebajo} yields 
\[
\| u\| \lqq \wW_{1}(u+X, X)\lqq \wW_{p}(u+X, X) \lqq \| u\|,
\]
which finishes the proof of~\eqref{e:shift}. 
 
\subsection{\textbf{The Gaussian characteristics of the limiting law for Gaussian Ornstein-Uhlenbeck perturbations}}\label{ss:OUCov}\hfill\\

\noindent Consider the Ornstein-Uhlenbeck process $(U(t))_{t\gqq 0}$ given by the solutions of the SDE
\[
U(t) = U_0 - \gamma \int_0^t U(s) \ud s + \sigma B(t), 
\]
where $U_0$ is independent of $(B(t))_{t\gqq 0}$ and $U_0 \stackrel{d}{=} \mathcal{N}(0, \frac{\sigma^2}{2\gamma})$. 
By linearity we have that the limiting law $\widetilde\gG = (\widetilde\gG_n)_{n\in \NN_0}$ is necessarily centered. 

\noindent
We now calculate the variance of $\widetilde\gG_n$. 

\begin{lemma}
For all $n\in \NN$ it follows that
\begin{align}
\EE[\widetilde \gG_n^2] 
&= \frac{\gamma^2 \EE[U_0^2]\Big(\frac{2}{\gamma + \nu}\Big)^{2n}}{\Big(1+ \sqrt{1+ \frac{4}{(\gamma + \nu)^2}}\Big)^{2n}}+\sigma^2 \int_0^\infty \Big(H_n(s) -\gamma \int_0^s H_n(u) e^{-\gamma (s-u)} \ud u\Big)^2 \ud s. 
\end{align}
In addition, 
\[
\int_0^\infty \Big(H_n(s) -\gamma \int_0^s H_n(u) e^{-\gamma (s-u)} \ud u\Big)^2 \ud s
\lqq  2\int_0^\infty H_n^2(s) (4\gamma +e^{-2\gamma s}) \ud s
\]
for all $n\in \NN$. 
\end{lemma}

\begin{proof} Note that
\begin{align}
A_n(t) &= \int_0^t H_n(t-s) \ud U(s)\\
&= - U_0 \gamma \int_0^t H_n(t-s) e^{-\gamma s} \ud s\\
&\,\quad - \gamma \sigma \int_0^t H_n(t-s) \Big(\int_0^s e^{-\gamma (s-u)} \ud B(u)\Big) \ud s + \sigma \int_0^t H_n(t-s) \ud B(s).  
\end{align}
Since $\EE[U_0]=0$, Fubini's Theorem implies $\EE[A_n(t)] = 0$. By It\^o's isometry we have
\begin{align}
\EE[A_n(t)^2] 
&= \EE[U_0^2] \gamma^2 \Big( \int_0^t H_n(t-s) e^{-\gamma s} \ud s\Big)^2 + \gamma^2 \sigma^2 
\EE\Big[\Big(\int_0^t H(t-s) \big(\int_0^s e^{-\gamma(t-u)} \ud B(u) \big) \ud s\Big)^2\Big] \\
&\qquad + \sigma^2 \EE\Big[\Big(\int_0^t H(t-s) \ud B(s)\Big)^2\Big]\\
&\qquad - 2 \sigma \gamma^2 \EE\Big[\Big(\int_0^t H_n(t-s) \Big(\int_0^s e^{-\gamma (s-u)} \ud B(u)\Big) \ud s\Big)\Big(\int_0^t H_n(t-s)\ud B(s)\Big)\Big], 
\end{align}
which can be simplified as follows 
\begin{align}
\EE[A_n(t)^2] = \gamma^2 \EE[U_0^2] \Big(\int_0^t H(t-s) e^{-\gamma s} \ud s \Big)^2 
+ \sigma^2 \int_0^t \Big( H(t-s) -\gamma e^{\gamma s} \int_s^t H(t-u) e^{-\gamma u} \ud u \Big)^2 \ud s.  
\end{align}
Sending $t\ra\infty$ we have 
\begin{align}
\EE[\widetilde \gG_n^2] 
 &= \gamma^2 \EE[U_0^2] \left(\frac{\frac{2}{\gamma + \nu}}{1+ \sqrt{1+ \frac{4}{(\gamma + \nu)^2}}} \right)^{2n}+ \sigma^2 \int_0^\infty \Big(H_n(s) -\gamma \int_0^s H_n(u) e^{-\gamma (s-u)} \ud u\Big)^2 \ud s.  
\end{align}
We estimate the second term on the right-hand side  of the preceding equality
\begin{align}
&\sigma^2 \int_0^\infty \Big(H_n(s) -\gamma \int_0^s H_n(u) e^{-\gamma (s-u)} \ud u\Big)^2 \ud s \\
\quad &= \sigma^2 \int_0^\infty \Big(H_n(s) e^{-\gamma s} + \gamma \int_0^s (H_n(s)-H_n(u)) e^{-\gamma (s-u)} \ud u\Big)^2 \ud s \\
\quad & \lqq 2 \sigma^2 \Big(\int_0^\infty H_n^2(s) e^{-2\gamma s} \ud s + 2\int_0^\infty \Big(\int_0^s (H(s) - H(u)) \gamma e^{-\gamma (s-u)} \ud u\Big)^2 \ud s\Big).   
\end{align}
We continue with the second term on the right-hand side of the preceding inequality (up to a constant factor $4\sigma^2$)
\begin{align}
&\int_0^\infty \Big(\int_0^s (H(s) - H(s-u)) \gamma e^{-\gamma u} \ud u\Big)^2 \ud s\\ 
&\qquad \lqq \int_0^\infty \int_0^s (H(s) - H(s-u))^2 \gamma e^{-\gamma u} \ud u \ud s \\
&\qquad  = \int_0^\infty \int_s^\infty (H(s) - H(s-u))^2 \gamma e^{-\gamma u} \ud s \ud u \\
&\qquad \lqq 4 \int_0^\infty \gamma e^{-\gamma u} \int_u^\infty H_n(s)^2 \ud s \ud u 
 = 4 \gamma \int_0^\infty H_n(s)^2 \ud s. 
\end{align}
\end{proof}

\subsection{\textbf{The case of Ornstein-Uhlenbeck perturbations with $\alpha$-stable driver}}\label{a:stableOU}\hfill\\

\noindent We consider $\EE[e^{\mathsf{i} u L(t)}] = e^{-t\sigma^\alpha |u|^\alpha}$, $u\in \mathbb{R}$ and 
\begin{equation}\label{e:OUestableinvariante}
\EE[e^{\mathsf{i} u U_0}] = e^{-tc_0 |u|^\alpha}\quad \textrm{ with }\quad c_0 = \frac{\sigma^\alpha}{\alpha \gamma}.
\end{equation}
We rewrite~\eqref{e: OUestable} as 
\begin{align}
U(t) &= U_0 - \gamma \int_0^t U(s) \ud s + L(t) = U_0 e^{-\gamma t} + \int_0^t e^{-\gamma (t-s)} \ud L(s). 
\end{align}
We calculate the $\alpha$-stable characteristics of the limiting law of the $n$-th component $\widetilde \gG_n$ of the limiting law $\widetilde \gG$. 

\begin{lemma} 
For all $n\in \NN$  it follows that 
\[
\EE[e^{\mathsf{i} u \widetilde \gG_n}] = e^{-(\sigma_n(\infty) |u|)^\alpha}, \quad u\in \RR,  
\]
where 
\begin{align}
\sigma_n^\alpha(\infty) = c_0\gamma^\alpha \Big|\int_0^\infty H_n(t-s) e^{-\gamma s} \ud s\Big|^\alpha + \sigma^\alpha  \int_0^\infty \big|H_n(s) - \int_0^s H_n(r) \gamma e^{-\gamma (s-r)} \ud r\big|^\alpha \ud s. 
\end{align}
In particular, for the absolute moment of order $0 < \theta < \alpha$ we have 
\begin{align}
\EE[|\widetilde \gG_n|^\theta] = (\sigma_n(\infty))^\theta 2^\theta  \frac{\Gamma(\frac{1+\theta}{2}) \Gamma(1-\frac{\theta}{\alpha})}{\sqrt{\pi} \Gamma(1-\frac{\theta}{2})}, 
\end{align}
where $\Gamma$ denotes the standard Gamma function.
\end{lemma}

\begin{proof}
Consider 
\begin{align}
&C_n(t) = \int_0^t H_n(t-s) \ud U(s)\\
&= - \gamma U_0 \int_0^t H_n(t-s) e^{-\gamma s} \ud s - \gamma \int_0^t H_n(t-s) \Big(\int_0^s e^{-\gamma(s-r)}\ud L(r)\Big) \ud s  + \int_0^t H_n(t-s) \ud L(s)\\
&=: L_{1,n}(t) + L_{2,n}(t) + L_{3,n}(t). 
\end{align}
We note that integration by parts gives 
\begin{align}
L_{2,n}(t) = - \int_0^t \Big(\int_s^t \gamma H_n(t-s) e^{-\gamma(r-s)} \ud r\Big) \ud L(s),   
\end{align}
which yields 
\begin{align}
L_{2,n}(t) + L_{3,n}(t) &= \int_0^t \Big(H_n(t-s) - \int_s^t H_n(t-r) e^{-\gamma(r-s)} \ud r\Big) \ud L(s), 
\end{align}
and which is independent from $L_{1,n}(t)$. Hence, we calculate with the help of~\cite[Lemma~17.1]{Sa} the characteristic function 
\begin{align}
\EE\Big[e^{\mathsf{i} u C_n(t)}] 
&= \EE\Big[e^{-\mathsf{i} U_0 u \gamma  \int_0^t H_n(t-s) e^{-\gamma s} \ud s}\Big] \cdot \EE\Big[e^{\mathsf{i} u( L_{2,n}+L_{3,n})}\Big]
= e^{-|u|^\alpha\sigma_{n}^\alpha(t)},
\end{align}
where 
\begin{align}
\sigma_n^\alpha(t) = c_0\gamma^\alpha \Big|\int_0^t H_n(t-s) e^{-\gamma s} \ud s\Big|^\alpha + \sigma^\alpha  \int_0^t \big|H_n(s) - \int_0^s H_n(r) \gamma e^{-\gamma (s-r)} \ud r\big|^\alpha \ud s. 
\end{align}
Sending $t\ra \infty$ we have that 
\[
\EE[e^{\mathsf{i} u \widetilde \gG_n}] = e^{-\sigma_n^\alpha(\infty) |u|^\alpha}, \quad u\in \RR.  
\]
Finally, we  calculate the absolute moment of order $0 < \theta < \alpha$ 
with the help of see~\cite[formula~(25.6)]{Sa} 
\begin{align}
\EE[|C_n(t)|^\theta] = (\sigma_n(t))^\theta 2^\theta  \frac{\Gamma(\frac{1+\theta}{2}) \Gamma(1-\frac{\theta}{\alpha})}{\sqrt{\pi} \Gamma(1-\frac{\theta}{2})}, 
\end{align}
where $\Gamma$ denotes the usual Gamma function. 
\end{proof}

\subsection{\textbf{The Ornstein-Uhlenbeck driven systems have laws in $\ell_2$}}\label{ap:subOUsystem}\hfill\\

\begin{lemma}\label{lem:sumavarian}
For any $\nu>0$ it follows that
\begin{equation}\label{e:Gaussiansplit}
\int_{0}^{\infty}\sum_{n=1}^{\infty} n^2 J^2_n(2r)\frac{e^{-2\nu r}}{r^2}\ud r<\infty.
\end{equation}
 \end{lemma}

\begin{proof}
Let $r\gqq 0$ and $n\in \mathbb{N}$ be fixed.
Recall the representation of the Bessel functions of the first kind  (see~\cite{AbramowitzStegung}, p.~376, Formula~9.6.18.)
\[
J_n(r)= \frac{(r/2)^n}{\sqrt{\pi}\Gamma(n+1/2)}\int_{-1}^{1} e^{\textrm{i}rt}(1-t^2)^{n-1/2}\ud t.
\]
By Jensen's inequality we obtain 
\begin{equation}\label{eq:formtri}
|J_n(2r)|\lqq \frac{r^n}{\sqrt{\pi}\Gamma(n+1/2)}\int_{-1}^{1} (1-t^2)^{n-1/2}\ud t.
\end{equation}
The change of variable $u=t^2$ with the help of 
Formula~6.2.1 and Formula~6.2.2 
in~pp.258 of~\cite{AbramowitzStegung}, and
the fact that $\Gamma(1/2)=\sqrt{\pi}$
gives
\[
\int_{-1}^{1} (1-t^2)^{n-1/2}\ud t
=2\int_{0}^{1} (1-t^2)^{n-1/2}\ud t=
\frac{\sqrt{\pi}\Gamma(n+1/2)}{\Gamma(n+1)}.
\]
Then we have
\begin{equation}\label{e:valchiq}
|J_n(2r)|\lqq \frac{r^n}{n!}.
\end{equation}
We note that this representation is useful only for small values of $r$. 
Recall the identity (see~\cite{AbramowitzStegung}, p.~363, Formula~9.1.76), 
\begin{align}\label{e:sumascuadrados}
J_0^2(s) + 2 \sum_{n=1}^\infty J_n^2(s)=1\quad \textrm{ for all }\quad s\gqq 0.
\end{align}
By~\eqref{e:sumascuadrados} we have 
\begin{equation}\label{e:valacot}
|J_n(s)|^2 \lqq\frac{1}{2}, \qquad n\in \NN,\, s\gqq 0.  
\end{equation}
By the monotone convergence theorem  and a subsequently split of the integral  along a monotonically growing and diverging sequence of positive numbers $(b_n)_{n\in \NN}$, $b_n\gqq 1$, which we determine in the sequel we obtain
\begin{align}
\int_{0}^{\infty}\sum_{n=1}^{\infty} n^2 J^2_n(2r)\frac{e^{-2\nu r}}{r^2}\ud r
&\lqq\sum_{n=1}^\infty \int_{0}^{\infty} n^2 J^2_n(2r)\frac{e^{-2\nu r}}{r^2}\ud r\\
&\lqq\sum_{n=1}^\infty \Big(\int_{0}^{b_n}+\int_{b_n}^\infty\Big) n^2 J^2_n(2r)\frac{e^{-2\nu r}}{r^2}\ud r= T_1 + T_2. 
\end{align}
We start with the estimate of $T_2$.  
By~\eqref{e:valacot} we have
\begin{align}
T_2 &= \sum_{n=1}^\infty \int_{b_n}^\infty n^2 J^2_n(2r)\frac{e^{-2\nu r}}{r^2}\ud r\lqq\sum_{n=1}^\infty n^2 \int_{b_n}^\infty e^{-2\nu r}\ud r\\
&\lqq\frac{1}{2\nu} \sum_{n=1}^\infty n^2 e^{-2 \nu b_n} \lqq\frac{\zeta(4)}{2\nu}  < \infty
\end{align}
for any $\nu>0$, where $\zeta(r) = \sum_{n=1}^\infty n^{-r}$ is the classical Riemann zeta function, 
if we choose 
\begin{equation}\label{eq:defbn}
b_n := \frac{3}{\nu} \ln(n)+1\quad \textrm{ for all }\quad n\in \mathbb{N}.
\end{equation}
We continue with the estimate of $T_1$. By~\eqref{e:valchiq} we have
\begin{align}
T_1 &\lqq\sum_{n=1}^{\infty} \frac{n^2}{(n!)^2} \int_{0}^{b_n} r^{2n-2} \ud r = \sum_{n=0}^{\infty} \frac{1}{(n!)^2} \int_{0}^{b_n} r^{2n} \ud r= \sum_{n=0}^{\infty} \frac{1}{(n!)^2} \frac{b_n^{2n+1}}{2n+1}.
\end{align}
We apply the root criterion to the $n$-th term with the help of Stirling's formula and obtain 
\begin{align}
\limsup_{n\ra\infty} \bigg(\frac{1}{(n!)^2} \frac{( \frac{3}{\nu} \ln(n)+1)^{2n+1}}{2n+1}\bigg)^\frac{1}{n}&= \limsup_{n\ra\infty} \bigg(\frac{e^{2}}{(n^{2}) (2\pi n)^\frac{1}{n}} \frac{( \frac{3}{\nu} \ln(n)+1)^{2+\frac{1}{n}}}{(2n+1)^\frac{1}{n}}\bigg)\\
&= \limsup_{n\ra\infty} e^2 \frac{(\frac{3}{\nu} \ln(n)+1)^2}{n^2} = 0,
\end{align}
which establishes the absolute convergence for any $\nu>0$. This shows formula~\eqref{e:Gaussiansplit}. 
\end{proof}

\noindent For the stable case, a similar split can be carried out.

\begin{lemma}\label{lem:stablesummability}
For any $\nu>0$ and $1 \lqq \theta < \alpha < 2$ it follows that 
\[
\sum_{n=1}^\infty  n^\theta \Big(\int_0^\infty \frac{|J_n(2r)|^\alpha e^{-\alpha \nu r}}{r^\alpha} \ud r\Big)^\frac{\theta}{\alpha} <\infty. 
\]
\end{lemma}
\begin{proof}
By~\eqref{e:valchiq} and~\eqref{e:valacot} and the subadditivity of the map $r\mapsto r^\frac{\theta}{\alpha}$ we may use the same kind of bounds
along a sequence positive, monotonically divergent sequence $(b_n)_{n\in \NN}$, $b_n\gqq 1$, 
 to be determined below,
\begin{equation}
\begin{split}
\sum_{n=1}^\infty  n^\theta \Big(\int_0^\infty \frac{|J_n(2r)|^\alpha e^{-\alpha \nu r}}{r^\alpha} \ud r\Big)^\frac{\theta}{\alpha}
&\lqq
\sum_{n=1}^\infty  n^\theta \bigg(\Big(\int_0^{b_n} \frac{|\frac{r^n}{n!}|^\alpha e^{-\alpha \nu r}}{r^\alpha} \ud r \Big)^\frac{\theta}{\alpha} + \Big(\int_{b_n}^\infty \frac{ e^{-\alpha \nu r}}{r^\alpha} \ud r \Big)^\frac{\theta}{\alpha} \bigg)\\
&\lqq 
\sum_{n=1}^\infty n^\theta \bigg(\frac{1}{(n!)^{\theta}}\Big(\int_0^{b_n} r^{\alpha (n-1)} \ud r \Big)^\frac{\theta}{\alpha} + \Big(\int_{b_n}^\infty e^{-\alpha \nu r} \ud r \Big)^\frac{\theta}{\alpha} \bigg)\\
&=  
\sum_{n=1}^\infty n^\theta  \bigg(\frac{1}{(n!)^{\theta}}\Big( \frac{(b_n)^{\alpha (n-1)+1}}{\alpha (n-1)+1}  \Big)^\frac{\theta}{\alpha} + \Big(\frac{e^{-\alpha \nu b_n}}{\alpha \nu} \Big)^\frac{\theta}{\alpha} \bigg)\\
&=  
\sum_{n=1}^\infty  \frac{n^\theta}{(n!)^{\theta}}  \frac{(b_n)^{\theta (n-1)+\frac{\theta}{\alpha}}}{(\alpha (n-1)+1)^\frac{\theta}{\alpha}}  + \sum_{n=1}^\infty  n^\theta \frac{e^{-\theta\nu b_n}}{(\alpha \nu)^\frac{\theta}{\alpha}} \\
&=  
\sum_{n=1}^\infty  \bigg(\frac{1}{(n-1)!}  \frac{(b_n)^{(n-1)+\frac{1}{\alpha}}}{(\alpha (n-1)+1)^\frac{1}{\alpha}}\bigg)^\theta  + \sum_{n=1}^\infty  \bigg(n \frac{e^{
-\nu b_n}}{(\alpha \nu)^\frac{1}{\alpha}}\bigg)^\theta \\
&\lqq
\sum_{n=0}^\infty  \bigg(\frac{(b_{n+1})^{n+\frac{1}{\alpha}}}{n!} \bigg)^\theta  + (\alpha \nu)^{-\frac{\theta}{\alpha}} \sum_{n=1}^\infty  \bigg(n e^{
-\nu b_n}\bigg)^\theta.
\end{split}
\end{equation}
The choice $(b_n)_{n\in \mathbb{N}}$ where $b_n$ is defined in~\eqref{eq:defbn} gives
\begin{equation}
\sum_{n=1}^\infty  n^\theta \Big(\int_0^\infty \frac{|J_n(2r)|^\alpha e^{-\alpha \nu r}}{r^\alpha} \ud r\Big)^\frac{\theta}{\alpha}\lqq
\sum_{n=0}^\infty  \bigg(\frac{(b_{n+1})^{n+\frac{1}{\alpha}}}{n!} \bigg)^\theta  + (\alpha \nu)^{-\frac{\theta}{\alpha}} e^{-\nu\theta} \sum_{n=1}^\infty  \frac{1}{n^{2\theta}}.
\end{equation}
For the first term of the right-hand side of the preceding inequality, the root criterion together with Stirling's formula yields 
\begin{align}
\limsup_{n\ra\infty} \bigg(\frac{(b_{n+1})^{n+\frac{1}{\alpha}}}{n!} \bigg)^{\frac{\theta}{n}}
&= \limsup_{n\ra\infty} \bigg(\frac{b_{n+1}}{(n!)^\frac{1}{n}} \bigg)^{\theta} = 0,
\end{align}
which implies the absolute convergence for any $1 \lqq\theta< \alpha \lqq2$. 
\end{proof}
 
\section*{\textbf{Acknowledgments}}
\noindent GB thanks the Academy of Finland, via the Matter and Materials Profi4 University Profiling Action, the Academy project No.~339228 and project No.~346306 of the Finnish Centre of Excellence in Randomness and STructures. 
GB would also like to thank the Instituto de Matem\'atica Pura e Aplicada (IMPA), Brazil, for support and hospitality during the 2023 Post-Doctoral Summer Program, where partial work on this paper was undertaken. The research of MAH was supported by the project INV-2023-162-2850 of Facultad de Ciencias at Universidad de los Andes. Part of this work was developed, whilst GB, MAH and IP visited CIMAT, M\'exico, in January 2023. All four authors thank CIMAT for the hospitality. 

\noindent
The authors are indebted with professor Jani Lukkarinen (Department of Mathematics and Statistics, University of Helsinki, Finland) for pointing out Stone's Theorem in Lemma~\ref{lem:Lyapunov}.

\section*{\textbf{Statements and declarations}}
\subsection*{\textbf{Availability of data and material}}
Data sharing not applicable to this article as no datasets were generated or analyzed during the current study.
\subsection*{\textbf{Conflict of interests}} The authors declare that they have no conflict of interest.
\subsection*{\textbf{Authors' contributions}}
All authors have contributed equally to the paper.

\end{document}